\documentclass{article}
\pdfoutput=1
\usepackage{graphicx,url,color,svgcolor,amsmath,amsfonts,amsthm}

\usepackage[ruled]{algorithm}
\usepackage{algorithmic}
\renewcommand{\algorithmicrequire}{{\textbf{Input:}}}
\renewcommand{\algorithmicensure}{{\textbf{Output:}}}

\usepackage{xspace}

\title{On Newton-Raphson iteration for multiplicative inverses modulo prime powers}
\author{Jean-Guillaume Dumas\thanks{Universit\'e Grenoble Alpes;
    Laboratoire Jean Kuntzmann, (umr CNRS 5224);
    700 avenue centrale, IMAG - CS 40700, 
    F-38058 Grenoble, France.
\href{mailto:Jean-Guillaume.Dumas@univ-grenoble-alpes.fr}{Jean-Guillaume.Dumas@univ-grenoble-alpes.fr}
  }}

\newtheorem{theorem}{Theorem}
\newtheorem{lemma}{Lemma}
\newtheorem{remark}{Remark}

\newcommand{\Z}{{\mathbb Z}}
\newcommand{\N}{{\mathbb N}}

\makeatletter
\usepackage{hyperref}
\hypersetup{
pdftitle={\@title},
pdfauthor={Jean-Guillaume Dumas},
breaklinks=true, 
colorlinks=true,
 linkcolor=blue,
 citecolor=darkred,
 urlcolor=darkblue,
}
\makeatother

\newcommand{\kindex}{\ensuremath{m}}
\begin{document}

\maketitle
\renewcommand{\algorithmicrequire}{{\textbf{Input:}}}
\renewcommand{\algorithmicensure}{{\textbf{Output:}}}
\makeatletter
\newcommand{\IFTHEN}[3][default]{\ALC@it\algorithmicif\ #2\
  \algorithmicthen\ #3\
  \ifthenelse{\boolean{ALC@noend}}{}{\algorithmicendif\ } \ALC@com{#1}}
\makeatother
\newcommand{\bigO}[1]{\ensuremath{\mathcal{O}\left(#1\right)}\xspace}
\newcommand{\LOGAND}{~\&\,}
\newcommand{\LOGANDIN}{~\&=\,}
\newcommand{\LOGORIN}{~|=\,}
\newcommand{\MULIN}{~*=\,}
\newcommand{\MOD}{~\%\,}
\newcommand{\MODIN}{~\%=\,}
\newcommand{\SUBIN}{~-=\,}
\newcommand{\ADDIN}{~+=\,}

\begin{abstract} 
  We study algorithms for the fast computation of modular inverses.
  Newton-Raphson iteration over $p$-adic numbers gives a recurrence relation computing modular inverse modulo $p^m$, that is logarithmic in $m$. 
  We solve the recurrence to obtain an explicit formula for the inverse. 
  Then we study different implementation variants of this iteration and show that our explicit formula is interesting for small exponent values but slower for large exponent, say of more than $700$ bits.
  Overall we thus propose a hybrid combination of our explicit formula and the best asymptotic variants. This hybrid combination yields then a constant factor improvement, also for large exponents. 
\end{abstract}
\section{Introduction}
The multiplicative inverse modulo a prime power is fundamental for the
arithmetic of finite rings, for instance at the initialization phase of
Montgomery's integer multiplication (see, e.g.,
\cite{Dusse:1990:eurocrypt,Arazi:2008:CMI} and references therein). 
It is also used, e.g., to compute homology groups in algebraic topology for
image pattern recognition \cite{jgd:2003:GAP}, mainly to improve the running
time of algorithms working modulo prime powers. Those can be used for the
computation of the local Smith normal
form~\cite{jgd:2001:JSC,Elsheikh:2012:ISSAC}, for instance in the context of
algebraic topology: there linear algebra
modulo $p^e$ can reveal torsion coefficients and inverses are required for
pivoting in Gaussian elimination or minimal polynomial synthesis (see, e.g.,
\cite[algorithm LRE]{jgd:2003:GAP} or \cite{Reeds:1985:SRS}).

Classical algorithms to compute a modular inverse uses the
extended Euclidean algorithm and Newton-Raphson iteration over p-adic fields, namely Hensel lifting \cite{Krishnamurthy:1983:padic}.
Arazi and Qi in \cite{Arazi:2008:CMI} lists also some variants adapted to the binary characteristic case that cut the result in lower and higher bits.

In the following, we give another proof of Arazi and Qi's logarithmic formula
using Hensel lifting. 
Then we derive an explicit formula for the inverse that generalizes to any prime
power. Finally, we study the respective performance of the different algorithms
both asymptotically and in practice and introduce a hybrid algorithm combining
the best approaches. 

\section{Hensel's lemma modulo \texorpdfstring{$p^{m}$}{p{\textasciicircum}m}}
For the sake of completeness, we first give here Hensel's lemma and
its proof from Newton-Raphson's iteration (see
e.g. \cite[Theorem 7.7.1]{Bach:1996:ANTEA} or \cite[\S 4.2]{Brent:2011:MCA} and
references therein). 
\begin{lemma}[Hensel]\label{lem:hensel} 
Let $p$ be a prime number, $\kindex \in \N$, $f \in \Z[X]$ and $r \in Z$
such that $f(r) = 0 \mod p^\kindex$. If $f'(r)\neq 0 \mod p^\kindex$ and $$t=-\frac{f(r)}{p^\kindex}f'(r)^{-1},$$ then
$s=r+tp^\kindex$ satisfies $f(s)=0 \mod p^{2\kindex}$.
\end{lemma}
\begin{proof}
Taylor expansion gives that
$f(r+tp^\kindex)=f(r)+tp^\kindex f'(r)+O(p^{2\kindex})$. Thus if
$t=-\frac{f(r)}{p^\kindex}f'(r)^{-1} $, 
the above equation becomes $f(s)=0 \mod p^{2\kindex}$.
\end{proof}

\section{Inverse modulo \texorpdfstring{$2^m$}{2{\textasciicircum}m}}
Now, in the spirit of \cite{Xenophontos:2010:fixed}, we apply this lemma to the
inverse function 
\begin{equation}\label{eq:invfun}
F_a(x)=\frac{1}{ax}-1
\end{equation}

\subsection{Arazi and Qi's formula}
We denote by an under-script $_L$ (resp. $_H$) the lower (resp. higher)
part in binary format for an integer.
From Equation~(\ref{eq:invfun}) and Lemma~\ref{lem:hensel} modulo $2^i$,
if $r=a^{-1} \mod 2^i$, then we immediately get
$$t=-\frac{ \frac{1}{ax}-1 }{2^i}\left(-\frac{1}{ax^2}\right)^{-1}.$$
In other words 
$t=\frac{1-ar}{2^i}r \mod 2^i$. Now let $a=b+2^i a_H \mod 2^{2i}$ so that we
also have $r=b^{-1} \mod 2^i$ and hence $r b = 1 + 2^i \alpha$ with
$0\leq \alpha < 2^i$. Thus
$a r = b r + 2^i r a_H = 1 +2^i( \alpha + r a_H)$
which shows that 
\begin{equation}\label{eq:arazi}
t =-( \alpha + r a_H)r
     \equiv -\left( \left(r b\right)_H + \left(r a_H\right)_L \right)r\mod 2^i
\end{equation}
The latter is exactly \cite[Theorem~1]{Arazi:2008:CMI} and yields the
following Algorithm~\ref{alg:arazi}, where the lower and higher parts of
integers are obtained via masking and shifting. 
\begin{algorithm}[htbp]
\caption{Arazi\&Qi Quadratic Modular inverse modulo $2^\kindex$}
\label{alg:arazi}
\begin{algorithmic}[1]
\REQUIRE $a \in \Z$ odd and $\kindex \in N$.
\ENSURE $U \equiv a^{-1} \mod 2^\kindex$.
\STATE $U=1$;
\FOR{($i=1$; $i<\kindex$; $i<<=1$)}
\STATE $b = a \LOGAND (2^i-1)$;\hfill\COMMENT{$b = a \mod 2^i$}
\STATE $t_1 = U * b$; $t_1 >>= i$;\hfill\COMMENT{$(rb)_H$}
\STATE $c = (a >> i) \LOGAND (2^i-1)$;\hfill\COMMENT{$a_H$}
\STATE $t_2 = (U * c) \LOGAND (2^i-1)$;\hfill\COMMENT{$(ra_H)_L$}
\STATE $t_1 \ADDIN t_2$;
\STATE $t_1 \MULIN U$; $t_1 \LOGANDIN (2^i-1)$;\hfill\COMMENT{$-t$}
\STATE $t_1 = 2^i-t_1$;\hfill\COMMENT{$t$}
\STATE $t_1 <<= i$;\hfill\COMMENT{$t 2^i$}
\STATE $U \LOGORIN t_1$;\hfill\COMMENT{$r+t 2^i$}
\ENDFOR
\STATE $U \LOGANDIN (2^\kindex-1)$;\hfill\COMMENT{$r \mod 2^\kindex$}
\RETURN $U$;
\end{algorithmic}
\end{algorithm}

\begin{lemma}\label{lem:ara} Algorithm~\ref{alg:arazi} requires $13 \lfloor
  \log_2(m) \rfloor+1$ arithmetic operations.
\end{lemma}



\subsection{Recurrence formula}
Another view of Newton-Raphson's iteration is to create a recurrence. 
Equation~(\ref{eq:invfun}) gives 
\begin{equation}\label{eq:rec}
\begin{split}
U_{n+1} &= U_n - \frac{\frac{1}{aU_n}-1}{-\frac{1}{aU_n^2}} = U_n-(aU_n-1)U_n \\
&= U_n(2-a U_n)
\end{split}
\end{equation}

This yields the loop of Algorithm~\ref{alg:hensel}, for the
computation of the inverse, 
see e.g. \cite{Krishnamurthy:1983:padic} or \cite[\S 2.4]{Brent:2011:MCA}.

\begin{algorithm}[htbp]
\caption{Hensel Quadratic Modular inverse}
\label{alg:hensel}
\begin{algorithmic}[1]
\REQUIRE $p\in \Z$ a prime, $a \in \Z$ coprime to $p$,  and $\kindex \in N$.
\ENSURE $U \equiv a^{-1} \mod p^\kindex$.
\STATE $U=a^{-1} \mod p$;\hfill\COMMENT{extended gcd}
\FOR{($i=2$; $i<\kindex$; $i<<=1$)}
\STATE\label{lin:sqr} $temp = U * U$;\hfill\COMMENT{$U_n^2$} 
\STATE $temp \MULIN a$;\hfill\COMMENT{$aU_n^2$} 
\STATE\label{lin:mod1} $temp \MODIN p^{i}$;\hfill\COMMENT{$temp \mod
  p^{i}$}
\STATE $U <<= 1$;\hfill\COMMENT{$2U_n$}
\STATE $U \SUBIN temp$;\hfill\COMMENT{$U_n(2-aU_n)$}
\ENDFOR 
\STATE $temp = U * U$;\hfill\COMMENT{$U_n^2$} 
\STATE $temp \MULIN a$;\hfill\COMMENT{$aU_n^2$} 
\STATE\label{lin:mod2} $temp \MODIN p^\kindex$;\hfill\COMMENT{$temp \mod
  p^{\kindex}$}
\STATE $U <<= 1$;\hfill\COMMENT{$2U_n$}
\STATE $U \SUBIN temp$;\hfill\COMMENT{$U_n(2-aU_n)$}
\STATE $U \MODIN p^\kindex$;\hfill\COMMENT{$U \mod p^\kindex$}

\RETURN $U$;
\end{algorithmic}
\end{algorithm}

\begin{lemma}\label{lem:rec} Algorithm~\ref{alg:hensel} is correct and requires 
$6 \lceil \log_2(m) \rceil + 2$ arithmetic operations.
\end{lemma}
\begin{proof} The proof of correctness is natural in view of the Hensel
lifting. First $U_0=a^{-1} \mod p$. Second, by induction, suppose 
$a \cdot U_n \equiv 1 \mod p^{k}$. Then $aU_n=1+\lambda p^k$ and $aU_{n+1} =
aU_n(2-aU_n)=(1+\lambda p^k)(2-1-\lambda p^k)=(1-\lambda^2p^{2k} \equiv 1 \mod
p^{2k})$. Finally $U_n \equiv a^{-1} \mod p^{2^n}$.
\end{proof}

\begin{remark} We present this algorithm for computations modulo $p^m$ but its
optimization modulo a power of $2$ is straightforward: replace the modular
operations of for instance lines~\ref{lin:mod1},~\ref{lin:mod2} etc. by a
binary masking: $x\LOGANDIN (2^i-1)$. 
\end{remark}

\begin{remark} It is important to use a squaring in line~\ref{lin:sqr}.
Indeed squaring can be faster than multiplication, in particular in the
arbitrary precision setting \cite{Zuras:1994:square}. 
In the case of Algorithm~\ref{alg:hensel}, the improvement over an algorithm of the
form $temp=2-a*U; temp \%= p^m; U*=temp;U\%=p^m;$ is of about $30\%$.
\end{remark}

\begin{remark}
Note that for Algorithms~\ref{alg:arazi} and~\ref{alg:hensel}, a large part of
the computation occur during the last iteration of the loop when $2^i$ is
closest to $2^m$. Therefore, a recursive version cutting in halves will be more
efficient in practice since the latter will be exactly done at $i=\kindex/2$
instead of at the largest power of $2$ lower than $\kindex$. Moreover this
improvement will take place at each recursion level. 
We thus give in the following the recursive version for Formula~(\ref{eq:rec}),
the one for a recursive version of Arazi\&Qi is in the same spirit. 

\begin{algorithm}[htbp]
\newcounter{prevalgorithm}
\setcounter{prevalgorithm}{\thealgorithm}
\renewcommand\thealgorithm{\theprevalgorithm'}
\addtocounter{algorithm}{-1}
\caption{Recursive Hensel}
\label{alg:rec}
\begin{algorithmic}[1]
\REQUIRE $p\in \Z$ a prime, $a \in \Z$ coprime to $p$, and $\kindex \in N$.
\ENSURE $r \equiv a^{-1} \mod p^\kindex$.
\IFTHEN{$\kindex == 1$}{\algorithmicreturn~$a^{-1} \mod p$;}\hfill\COMMENT{ext. gcd}
\STATE $h=\lceil\frac{\kindex}{2}\rceil$
\STATE\label{lin:modph} $b = a \MOD p^h$;\hfill\COMMENT{$b = a \mod p^h$}
\STATE $r=$RecursiveHensel$(p,b,h)$;
\STATE $temp = r * r$;\hfill\COMMENT{$r^2$} 
\STATE $temp \MULIN a$;\hfill\COMMENT{$a r^2$} 
\STATE\label{lin:modpm} $temp \MODIN p^{m}$;\hfill\COMMENT{$temp \mod p^{m}$}
\STATE $r <<= 1$;\hfill\COMMENT{$2r$}
\STATE $r \SUBIN temp$;\hfill\COMMENT{$r(2-ar)$}
\STATE $r \MODIN p^\kindex$;\hfill\COMMENT{$r \mod p^{\kindex}$}
\RETURN $r$;
\end{algorithmic}
\end{algorithm}
\end{remark}

\subsection{Factorized formula}
We now give an explicit formula for the inverse by solving the
preceding recurrence relation, first in even characteristic.\\

We denote by $H_n=aU_n$ a new sequence, that
satisfies $H_{n+1}=H_n(2-H_n)$.
With $H_0=a$ we get $H_1=a(2-a)=2a-a^2=1-(a-1)^{2^1}$,
by induction, supposing that $H_n=1-(a-1)^{2^i}$, we get
\begin{align*}
H_{n+1}&=\left(1-(a-1)^{2^n}\right)\left(2-1+(a-1)^{2^n}\right)\\
& = 1^2-\left((a-1)^{2^n}\right)^2 =1-(a-1)^{2^{n+1}}
\end{align*}
Using the remarkable identity, this in turn yields:
$$H_{n} = a(2-a)\prod_{i=1}^{n-1}\left(1+(a-1)^{2^i}\right);$$ therefore, with
$U_0=1$ and $U_1=2-a$ we have that  
\begin{equation}\label{eq:expl}
U_{n}= (2-a)\prod_{i=1}^{n-1}\left(1+(a-1)^{2^i}\right)
\end{equation}

The latter equation gives immediately rise to the following Algorithm~\ref{alg:expl}. 

\begin{algorithm}[htbp]
\caption{Explicit Quadratic Modular inverse modulo $2^\kindex$}
\label{alg:expl}
\begin{algorithmic}[1]
\REQUIRE $a \in \Z$ odd and $\kindex \in N$.
\ENSURE $U \equiv a^{-1} \mod 2^\kindex$.
\STATE Let $s$ and $t$ be such that $t$ is odd and $a=2^st+1$;
\STATE $U=2-a$;
\STATE $amone = a-1$;
\FOR{($i=1$; $i<\frac{\kindex}{s}$; $i<<=1$)}
\STATE $amone \MULIN amone$; \hfill\COMMENT{square: $(a-1)^{2^i}$}
\STATE $amone \LOGANDIN (2^\kindex-1)$;\hfill\COMMENT{$(a-1)^{2^i} \mod 2^\kindex$}
\STATE $U \MULIN (amone+1)$; 
\STATE $U \LOGANDIN (2^\kindex-1)$;\hfill\COMMENT{$U \mod 2^\kindex$}
\ENDFOR
\RETURN $U$;
\end{algorithmic}
\end{algorithm}

\begin{lemma} Algorithm~\ref{alg:expl} is correct and 
requires $5 \lfloor\log_2(\frac{m}{s})\rfloor + 2$ arithmetic operations.
\end{lemma}
\begin{proof}
Modulo $2^m$, $a$ is invertible if and only if $a$ is odd, so that
$a=2^st+1$ and therefore, using Formula~(\ref{eq:expl}), we get
$a U_n = H_n = 1-(a-1)^{2^n}=1-(2^st)^{2^n} \equiv 1 \mod 2^{s2^n}.$
Thus,
$U_{\lceil\log_2(\frac{m}{s})\rceil} \mod 2^m \equiv a^{-1} \mod 2^m.$
\end{proof}

There are two major points to remark with this variant:
\begin{enumerate}
\item It performs fewer operations than previous algorithms.
\item It must compute with the full $p$-adic development (modulo operations are
  made modulo $2^m$ and not $2^i$.
\end{enumerate}
Therefore we will see that this algorithm has a worse asymptotic complexity
but is very efficient in practice for small exponents.

\subsection{Generalization modulo any prime power}
The formula generalizes directly for any prime power:
\begin{theorem}
Let $p$ be a prime number, $a$ coprime to $p$ and $b=a^{-1} \mod p$ is the
inverse of $a$ modulo $p$.
Let also $V_n$ be the following sequence:
\begin{equation}
\begin{cases}
V_0 &= b \equiv a^{-1} \mod p,\\
V_n &= b(2-ab)\prod_{i=1}^{n-1}\left(1+(ab-1)^{2^i}\right)
\end{cases}
\end{equation}
Then $V_n \equiv a^{-1} \mod p^{2^n}$.
\end{theorem}
\begin{proof}
The proof is similar to that of Lemma~\ref{alg:expl} and follows also from
Hensel's lemma.
From the analogue of Equation~(\ref{eq:expl}), we have $a\cdot
V_n=1-(ab-1)^{2^n}$. 
Now as $a\cdot b = 1 + \lambda p$, by the definition of $b$ we have 
$a\cdot V_n=1-(ab-1)^{2^n} = 1 - (\lambda p)^{2^n} \equiv 1 \mod p^{2^n}$. 
\end{proof}

\section{Complexity analysis over arbitrary precision}\label{sec:ana}
We provide here the equivalents of the complexity results of the previous
section but now for arbitrary precision: the associated binary
complexity bounds for the different algorithms are given here with classical
arithmetic operations on integer (i.e. without fast variants like Karatsuba or
DFT). 
We thus now suppose that masking and shifting as well as addition are linear
and that multiplication is quadratic (\bigO{2m^2} operations to multiply to
elements of size $m$).
\begin{lemma}\label{lem:arazi} 
Using classical arithmetic, Algorithm~\ref{alg:arazi} requires
$$\bigO{2\kindex^2+10\kindex} ~~\text{binary operations.}$$
\end{lemma}
\begin{proof} Following the algorithm, the complexity bound becomes 
$$\bigO{\kindex+\sum_{j=1}^{\log_2(\kindex)-1} 3 \cdot 2(2^j)^2 + 1 (2^j) }=\bigO{2\kindex^2+10\kindex}.$$
\end{proof}
\smallskip
\begin{lemma}\label{lem:henarb} 
Using classical arithmetic in even characteristic modulo $2^m$, Algorithm~\ref{alg:hensel} requires $$\bigO{\frac{16}{3}\kindex^2+9\kindex} ~~\text{binary operations.}$$
\end{lemma}
\begin{proof} Following the algorithm, the complexity bound becomes
  $$\bigO{2\cdot 2\kindex^2+5\kindex+\sum_{j=2}^{\log_2(\kindex)-1} 2 \cdot 2(2^j)^2 + 4(2^j) }=\bigO{ \frac{16}{3}\kindex^2+9\kindex}.$$
\end{proof}
\smallskip
\begin{lemma}\label{lem:explarb} 
Using classical arithmetic, Algorithm~\ref{alg:expl} requires
$$\bigO{\left(4\kindex^2+2\kindex\right)\lfloor\log_2(\kindex)\rfloor} ~~\text{binary operations.}$$
\end{lemma}
\begin{proof} Similarly, here we have 
$$\bigO{\sum_{j=1}^{\log_2(\kindex)-1} 2 \cdot 2\kindex^2 + 2\kindex}=\bigO{\left(4\kindex^2+2\kindex\right)\log_2(m)}.$$\end{proof}

\section{Experimental comparisons}
The point of the classical Newton-Raphson algorithms (as well as Arazi and Qi's
variant) is that it works with modular computations of increasing sizes, whereas
the explicit formula requires to work modulo the highest size from the
beginning. 
On the one hand we show next that this gives an asymptotic advantage the recurring relations. On the other hand, in practice, the explicit formula enables
much faster performance for say cryptographic sizes.
All experiments have been done on an Intel Xeon W3530, 2.8 GHz, running linux
debian\footnote{The source code for the experiments is available on
  \url{http://ljk.imag.fr/membres/Jean-Guillaume.Dumas/Software/InvModTwoK}. It
  uses GMP version 5.0.5 (\url{http://gmplib.org}), with givaro-3.7.2 C++
  wrappers (\url{http://givaro.forge.imag.fr})}
\subsection{Over word-size integers}
Using word-size integers, the many masking and shifting required by recurring
relations do penalize the performance, where the simpler Algorithm~\ref{alg:expl} is on average $26\%$ faster on a
standard desktop PC, as shown on Figure~\ref{fig:uint}. Differently, Arazi and
Qi's variant suffers from the manipulations required to extract the low and high
parts of integers.
\begin{figure}[htb]\center\vspace{-2pt}
\includegraphics[width=\columnwidth]{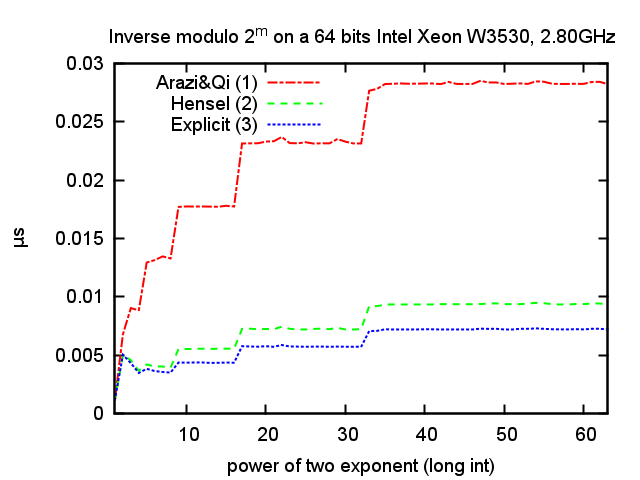}
\caption{Modular inverse on 64 bits machine words}\label{fig:uint}\vspace{-2pt}
\end{figure}

\subsection{Over arbitrary precision arithmetic}
From Lemma~\ref{lem:explarb}, we see that the explicit formula adds a logarithmic factor, asymptotically. 
In practice, Figure~\ref{fig:gmp} shows that using
GMP\footnotemark[1], the asymptotic behavior of Algorithm~\ref{alg:arazi} becomes predominant only for integers with more than $1200$ bits. 
For the Newton-Raphson iteration the asymptotic behavior of Algorithm~\ref{alg:hensel} becomes predominant even sooner, for integers with about $640$
 bits. 
Below that size, Algorithm~\ref{alg:expl} is better.
\begin{figure}[htb]\center\vspace{-2pt}
\includegraphics[width=\columnwidth]{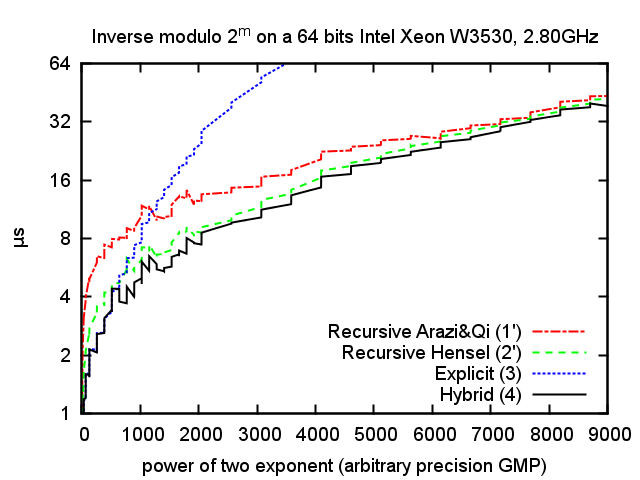}
\caption{Modular inverse on arbitrary precision integers}\label{fig:gmp}\vspace{-2pt}
\end{figure}

\begin{remark} 
  Now, in \cite{Krishnamurthy:1983:padic,Xenophontos:2010:fixed},
  the recurrence relation from~(\ref{eq:rec}), is extended\footnote{this is to
    be compared with explicit Formula~(\ref{eq:expl}), $V_{n+1} = V_n
    (1+(1-ab)^{2^n})$, where the computation is done with the first inverse $mod
    p$ (recall that $b\equiv a^{-1} \mod p$), where in the classical setting the
    computation is done with the inverse so far: $X_n$}
  to $X_{n+1} = \frac{1-(1-a X_n)^r}{a}$ for a fixed $r$. 
  This allows a faster convergence, by a factor of $\log_2(r)$. 

  Unfortunately the price to pay
  is to compute a $r$-th power at each iteration (instead of a single square),
  which could be done, say by recursive binary squaring, but at a price of
  $\log_2(r)$ squarings and between $0$ and $\log_2(r)$ multiplications. Overall
  there would be no improvement in the running time.
\end{remark}

\subsection{Hybrid algorithm}

With the thresholds of Figure~\ref{fig:gmp} and from the previous algorithms, 
we can then use a classical hybrid strategy which is better than
all of them everywhere:
for small exponents it uses the explicit formula of Algorithm~\ref{alg:expl};
then for larger exponents:
\begin{enumerate}
\item it starts to compute the inverse recursively at half the initial exponent;
\item then, to lift the inverse modulo the double exponent, it uses the classical Hensel formula of Equation~(\ref{eq:rec}), and
  switches to Arazi\&Qi formula of Equation~(\ref{eq:arazi}), only for exponents
  larger than $9000$ bits. 
\end{enumerate}
Actually, the lift switches back to Hensel formula after $10^6$ bits: indeed on
the used computer quasi linear multiplication via FFT comes into play in GMP and
the analysis of Section~\ref{sec:ana} is not relevant anymore.
The obtained algorithm is on average $21\%$ times faster than any
other direct lifting alone as shown with the curve (4) of
Figure~\ref{fig:gmp} (recall that on Figure~\ref{fig:gmp} ordinates are
presented in a logarithmic scale) and also on the ratios of
Figure~\ref{fig:ratio}.
\begin{figure}[htb]\center\vspace{-2pt}
\includegraphics[width=\columnwidth]{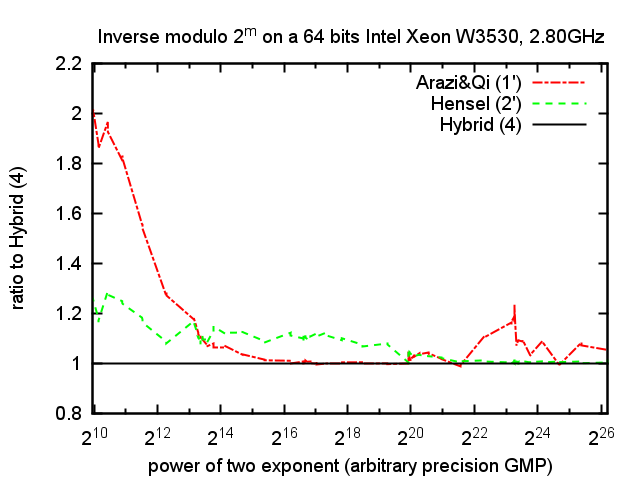}
\caption{Ratios of modular inverse lifting algorithms over the hybrid method}\label{fig:ratio}\vspace{-2pt}
\end{figure}

\section{Conclusion}
We have studied different variants of Newton-Raphson's iteration over
p-adic numbers to compute the inverse modulo a prime power.
We have shown that a new explicit formula can be up $26\%$ times faster in
practice than the recursive variants for small exponents.
Asymptotically, though, the latter formula suffers from a supplementary
logarithmicfactor in the power (or a doubly logarithmic factor in the prime
power) that makes it slower for large arbitrary precision integers. 
However, using each one of the best two algorithms in their respective regions
of efficiency, we were able to use a hybrid strategy with improved performance
of $21\%$ on average at any precision.  

More studies are to be made for the respective behavior of the algorithms in odd
characteristic. Indeed there bit masking is replaced by extended euclidean
algorithms variants and their respective performance could be different.

\section*{Acknowledgment} 
Many thanks to Christoph Walther, who found two typographic mistakes 
in the published version of Algorithm~\ref{alg:rec} ($2^h$ where $p^h$ was
correct, line~\ref{lin:modph} and $p^h$ where $p^m$ was correct,
line~\ref{lin:modpm})~\cite{Walther:2018:cav}.


\bibliographystyle{plain}
\bibliography{invmodpk}

\begin{thebibliography}{10}

\bibitem{Arazi:2008:CMI}
O.~Arazi and Hairong Qi.
\newblock On calculating multiplicative inverses modulo $2^{m}$.
\newblock {\em IEEE Transactions on Computers}, 57(10):1435--1438, October
  2008.

\bibitem{Bach:1996:ANTEA}
Eric Bach and Jeffrey Shallit.
\newblock {\em Algorithmic Number Theory: Efficient Algorithms}.
\newblock MIT press, 1996.

\bibitem{Brent:2011:MCA}
Richard~P. Brent and Paul Zimmermann.
\newblock {\em Modern computer arithmetic}, volume~18 of {\em Cambridge
  monographs on applied and computational mathematics}.
\newblock Cambridge University Press, Cambridge, UK, 2011.

\bibitem{jgd:2003:GAP}
{Jean-Guillaume} {Dumas}, {Frank} {Heckenbach}, {B. David} {Saunders}, and
  {Volkmar} {Welker}.
\newblock Computing simplicial homology based on efficient {Smith} normal form
  algorithms.
\newblock In Michael Joswig and Nobuki Takayama, editors, {\em Algebra,
  Geometry and Software Systems}, pages 177--206. Springer, 2003.

\bibitem{jgd:2001:JSC}
Jean-Guillaume Dumas, B.~David Saunders, and Gilles Villard.
\newblock On efficient sparse integer matrix {Smith} normal form computations.
\newblock {\em Journal of Symbolic Computation}, 32(1/2):71--99, July--August
  2001.

\bibitem{Dusse:1990:eurocrypt}
Stephen~R. Duss\'e and Burton~S. {Kaliski~Jr.}
\newblock A cryptographic library for the {Motorola} {DSP56000}.
\newblock In {\em EUROCRYPT '90, Denmark, May 21-24, 1990, Proceedings}, volume
  473 of {\em Lecture Notes in Computer Science}, pages 230--244, 1990.

\bibitem{Elsheikh:2012:ISSAC}
Mustafa Elsheikh, Mark Giesbrecht, Andy Novocin, and B.~David Saunders.
\newblock Fast computation of {Smith} forms of sparse matrices over local
  rings.
\newblock In Joris van~der Hoeven and Mark van Hoeij, editors, {\em
  {ISSAC}'2012, Proceedings of the 2012 ACM International Symposium on Symbolic
  and Algebraic Computation, Grenoble, France}, pages 146--153, July 2012.

\bibitem{Xenophontos:2010:fixed}
Michael Knapp and Christos Xenophontos.
\newblock Numerical analysis meets number theory: using root finding methods to
  calculate inverses mod $p^n$.
\newblock {\em Applicable Analysis and Discrete Mathematics}, 4(1):23--31,
  2010.

\bibitem{Krishnamurthy:1983:padic}
E.V. Krishnamurthy and Venu~K. Murthy.
\newblock Fast iterative division of p-adic numbers.
\newblock {\em IEEE Transactions on Computers}, C-32(4):396--398, April 1983.

\bibitem{Reeds:1985:SRS}
James~A. Reeds and Neil J.~A. Sloane.
\newblock Shift-register synthesis (modulo $m$).
\newblock {\em SIAM Journal on Computing}, 14(3):505--513, August 1985.

\bibitem{Walther:2018:cav}
Christoph Walther.
\newblock Formally verified {Montgomery} multiplication.
\newblock In {\em 30th International Conference on Computer Aided Verification
  (CAV'18), Oxford, UK}, 2018.

\bibitem{Zuras:1994:square}
Dan Zuras.
\newblock More on squaring and multiplying large integers.
\newblock {\em IEEE Transactions on Computers}, 43(8):899--908, August 1994.

\end{thebibliography}
\end{document}